\documentclass[journal]{IEEEtran}
\usepackage{graphicx} 
\graphicspath{{./figures/}}
\usepackage{cite}                               
\usepackage{amsmath, amsfonts,amssymb,amsthm}
\usepackage{multicol,multirow}
\usepackage[usenames, dvipsnames]{xcolor}
\usepackage{soul}
\usepackage{epstopdf}
\usepackage{etoolbox}

\usepackage{algorithmic}
\usepackage{textcomp}
\def\BibTeX{{\rm B\kern-.05em{\sc i\kern-.025em b}\kern-.08em
		T\kern-.1667em\lower.7ex\hbox{E}\kern-.125emX}}

\newtheorem{theorem}{Theorem}
 
\newtheorem{assumption}{Assumption}

\newtheorem{lemma}{Lemma}

\begin{document}
\bstctlcite{IEEEexample:BSTcontrol}
\title{{Artificial Delay Based ARC of a Class of Uncertain EL Systems with Only Position Feedback}}

\author{Spandan~Roy,~Indra~Narayan~Kar,~\IEEEmembership{Senior~Member,~IEEE}, ~Jinoh Lee,~\IEEEmembership{Member,~IEEE},\\
	Nikos Tsagarakis,~\IEEEmembership{Member,~IEEE}~and~Darwin~G.~Caldwell, \IEEEmembership{Member, IEEE} 
	\thanks{S. Roy and I. N. Kar are with the Department
		of Electrical Engineering, Indian Institute of Technology Delhi, New Delhi,
		India {(e-mail:} sroy002@gmail.com, ink@ee.iitd.ac.in).
	}
	\thanks{J. Lee, N. Tsagarakis and D. G. Caldwell are with the Advanced Robotics Department,
		Fondazione Istituto Italiano di Tecnologia (IIT), Genoa,
		Italy {(e-mail: \{jinoh.lee, nikos.tsagarakis, Darwin.Caldwell\}}@iit.it).}
} 

\maketitle

\begin{abstract}
	In this paper, the tracking control problem of {an Euler-Lagrange system is addressed with regard to parametric uncertainties,} and an adaptive-robust control strategy, {{christened} 
		Time-Delayed Adaptive Robust Control (TARC), is presented}. TARC approximates the unknown dynamics through the time-delayed estimation, and the adaptive-robust control provides robustness against the approximation error. The novel adaptation law of TARC, in contrast to the conventional adaptive-robust control methodologies, requires neither complete model of the system nor any knowledge of predefined uncertainty bounds to compute the switching gain, and circumvents the over- and under-estimation problems of the switching gain. {Moreover, TARC only utilizes position feedback and approximates the velocity and acceleration terms from the past position data. The adopted state-derivatives estimation method in TARC avoids any explicit requirement of external low pass filters for the removal of measurement noise. A new stability notion in continuous-time domain is proposed considering the time delay, adaptive law, and state-derivatives estimation which in turn provides a selection criterion for gains and sampling interval of the controller}.
\end{abstract}

\section*{Nomenclature}
\begin{IEEEdescription}[\IEEEusemathlabelsep\IEEEsetlabelwidth{$\lambda _{\min}(\bullet)$}]
	\item[$\varphi_{h}$] $\varphi(t)$ is delayed by an amount $h$ as $\varphi(t-h)$.
	\item[$\lambda _{\min}(\bullet)$]  Minimum eigen value of $(\bullet)$.
	\item[$|| (\bullet) \Vert$] Euclidean norm of $(\bullet)$.
	\item[$\mathbf{I}$] Identity matrix of proper dimension.
	\item[$\text{sgn}$] Standard signum function.
\end{IEEEdescription}

\section{Introduction}
\subsection{Background and Motivation}\label{intro}
\IEEEPARstart{D}{esign} of an efficient controller for nonlinear systems subjected to parametric and nonparametric uncertainties has always been a challenging task. Among many other approaches, adaptive control and robust control are the two popular control strategies that researchers have extensively employed while dealing with uncertain nonlinear systems. {While adaptive control estimates the unknown system parameters and controller gains online \cite{Ref:2}, a robust controller such as classical Sliding Mode Control (SMC) provides robustness against system uncertainties within a predefined uncertainty bound \cite{Ref:3}. However, defining a prior uncertainty bound is not always possible due to the unmodelled dynamics and external disturbances. Moreover, to increase the operating region of SMC, often higher values of uncertainty bounds are assumed. This in turn leads to overestimation of switching gain and reduces controller accuracy of the robust controller~\cite{Ref:4}.} 

Recently, the global research is reoriented towards adaptive-robust control (ARC), where attempts are made to reap the benefits of both the adaptive and robust control methods by applying them simultaneously. The series of publications \cite{Ref:2}, \cite{Ref:13,Ref:16,Ref:17,Ref:18,Ref:19} regarding ARC, estimate the uncertain terms online based on a predefined projection function. However, usage of projection function requires upper and lower bounds of individual uncertain parameters, which is not always possible in practice. The adaptive sliding mode control (ASMC) based designs \cite{Ref:20,Ref:20-1,Ref:20-2,Ref:20-3,Ref:32,Ref:21,Ref:22} adapt the switching gain online without any predefined knowledge of the bound of uncertainty. The benefits of such designs over \cite{Ref:2}, \cite{Ref:13,Ref:16,Ref:17,Ref:18,Ref:19} are: (i) rather than adapting to each of the multiple uncertain system parameters, it is sufficient to adapt only a single parameter, the switching gain and (ii) knowledge of the uncertainty bound is not a prerequisite. Nevertheless, the adaptive laws reported in \cite{Ref:20,Ref:20-1,Ref:20-2,Ref:20-3,Ref:32} make the switching gain a monotonically increasing function. Thus, the controllers become susceptible to very high switching gain and consequent chattering \cite{Ref:33}. 

{The ASMC proposed in \cite{Ref:21,Ref:22} overcome {the} monotonic nature of switching gain}. However, the adaptive laws of \cite{Ref:21,Ref:22} involve a predefined threshold value; it is worth to notice, until the threshold value is achieved, the switching gain may still be increasing (resp. decreasing) even if the tracking error decreases (resp. increases) and thus creates \textit{overestimation} (resp. \textit{underestimation}) problem of the switching gain. While the underestimation problem compromises control accuracy by providing lower values of the switching gain than the required amount, the overestimation problem demands excessive control input by providing higher values of the switching gain than the required amount. Especially, the overestimation problem may invite chattering if the switching gain becomes too high \cite{Ref:4,Ref:33}. Furthermore, 
the controllers in \cite{Ref:21,Ref:22} assume that the nominal absolute values of all uncertain parameters are greater than the corresponding perturbation terms. Such assumption necessitates complete modelling of the system which is not possible in the presence of unmodeled dynamics. 

{To avoid complete prior knowledge of the system model, researchers have applied 
	Time-Delayed Control (TDC)\cite{Ref:6,Ref:6-2,Ref:26,Ref:new1} to approximate uncertain system dynamics. 
	The time-delayed estimation (TDE) method in TDC approximates the lumped system uncertainty by only using control input and state information of the immediate past time instant and the design process does not require expertise knowledge.} In spite of this, the unattended approximation error, commonly termed as TDE error causes detrimental effect to the performance of the closed loop system and its stability. In this front, a few works have been carried out to tackle the TDE error which include internal model \cite{Ref:7}, ideal velocity feedback \cite{Ref:9,Ref:9_1}, nonlinear damping \cite{Ref:10}, and conventional SMC \cite{Ref:12,Ref:34}. 

{It is {worthwhile to notice} that Fuzzy logic based adaptive control techniques have also been exploited \cite{Ref:faltu}. {Nevertheless}, as observed in \cite{Ref:jin}, classical TDC has less computation burden compared to such black-box technique. 
However, all the {TDE-}based works necessitate state-derivatives feedbacks for Euler-Lagrange (EL) systems.}
In the absence of {the state-derivative terms (more specifically,} velocity and acceleration feedback in TDE-based {controllers} for EL systems), numerical differentiation for state-derivatives, computed from noisy state data (i.e., position information for EL systems) often invites considerable measurement error, which degrades controller performance \cite{Ref:7,Ref:6-1}. {Nevertheless, the effect of numerical approximations of the state-derivatives on the stability of the overall system is not considered in }\cite{Ref:7,Ref:6-1}. Furthermore, as observed in \cite{Ref:7}, the usage of a low pass filter (LPF) in TDC to mitigate the effect of measurement noise has pervasive effects on system stability as well as controller performance. {Yet again, the choice of the time-delay value, i.e., sampling time, for TDE-based controllers and its effect on the stability of the system are important issues. The authors in} \cite{Ref:dis_tdc_1,Ref:dis_TDC_2} {report the choice of sampling time and the corresponding stability analysis in discrete time domain. However, all the reported TDE based controllers} \cite{Ref:jin,Ref:6,Ref:6-2,Ref:26,Ref:new1,Ref:7,Ref:9,Ref:9_1,Ref:10,Ref:12,Ref:34} {select the controller gains and the delay value independently which is conservative in nature, and the impact of such selections on the overall system stability has not been sufficiently analyzed in the continuous-time domain}. Moreover, \cite{Ref:12,Ref:34} require predefined bound on TDE error; although the ARC laws in \cite{Ref:21,Ref:22} do not require any prior knowledge of the uncertainty bound, its adaptive laws suffer from the over- and under-estimation problems. Hence, it is imperative to formulate a controller which can address the aforementioned individual issues of both the conventional continuous-time TDC and TDE based controllers as well as ARCs.

\subsection{Contributions}\label{sec 1.2}
{ In this article, an adaptive-robust control strategy, Time-Delayed Adaptive Robust Control (TARC) has been formulated for a class of uncertain EL systems which provides an integrated and comprehensive solution to the existing issues of TDE based controllers \cite{Ref:jin,Ref:6,Ref:6-2,Ref:26,Ref:new1,Ref:7,Ref:9,Ref:9_1,Ref:10,Ref:12,Ref:34} as well as ARCs \cite{Ref:21,Ref:22}. The proposed TARC approximates the unknown system dynamics by the TDE method and provides robustness against the TDE error by an adaptive-robust control. 
	
	The main contribution of the proposed TARC is its novel adaptive law which does not involve any threshold values and thus, alleviates the over- and under-estimation problems of the switching gain compared to \cite{Ref:21,Ref:22}. Additionally, to enhance the practical applicability of TARC, the state-derivatives estimation technique \cite{Ref:30} is adopted in TARC which avoids any explicit requirement of velocity and acceleration information for an EL system. The state-derivatives estimation procedure \cite{Ref:30} itself has noise suppressing capability which eliminates any requirement of external LPFs; thus, TARC can avoid the pervasive effect of LPF commonly seen on the stability of TDE-based control \cite{Ref:7}. This paper also offers complete continuous-time domain stability analysis of the proposed TARC considering the time delay component in the controller, the adaptive law, and state-derivatives estimation. 
	It indeed establishes an analytical procedure to give a selection criterion for controller gains and sampling time. }

\subsection{Organization}
The rest of the paper is organized as follows: 
In Section~\ref{sec: 2}, the design issues of TDC is first clarified and a new stability analysis of TDC is provided. This is followed by the proposed adaptive-robust control law. {The stability analysis and parameter selection of TARC are separately provided in Section~III. }Section IV concludes the entire work.

\section{Controller Design}\label{sec: 2}

\subsection{Time-Delayed Control and A New Stability Analysis}\label{sec: 2.1}
This subsection revisits the robust control scheme known as TDC for uncertain Euler-Lagrange (EL) systems 
and provides a new stability analysis of TDC in the sense of Lyapunov. This subsection provides a building block for the proposed controller in the next subsection \ref{sec: 2.2}. 

In general, an EL system with second order dynamics, devoid of any delay, can be written as
\begin{equation}\label{sys}
\mathbf{ M(q)}\ddot{\mathbf q}+\mathbf N(\mathbf q,\dot{\mathbf q})=\boldsymbol\tau,
\end{equation}
where $\mathbf{q}(t)\in\mathbb{R}^{n}$ denotes position for EL system, $\boldsymbol\tau\in\mathbb{R}^{n}$ is the control input, $\mathbf{M(q)}\in\mathbb{R}^{n\times n}$  is the mass/inertia matrix and ${\mathbf N(\mathbf q,\dot{\mathbf q})}\in\mathbb{R}^{n}$ denotes combination of other system dynamics terms based on system properties such as Coriolis, gravitational, friction, damping forces. The control input is defined to be
\begin{equation}\label{input}
\boldsymbol\tau={\hat{\mathbf M} \mathbf u+\hat{\mathbf N}},
\end{equation}
where $\mathbf{u}$ is the auxiliary control input; $\mathbf{\hat{N}}$ is the nominal value of $\mathbf{N}$ and {$\mathbf{\hat{M}}$ is a constant matrix selected from the nominal knowledge of $\mathbf M$ \cite{Ref:33,Ref:6,Ref:6-2}}. To reduce the modelling effort of the systems, $\mathbf{\hat{N}}$ can be approximated from the input-output data of previous instant using the time-delayed estimation (TDE) method \cite{Ref:6,Ref:6-2,Ref:33} and the system definition (\ref{sys}) as

\begin{equation}\label{approx}
{\hat{\mathbf N}( \mathbf q,\dot{\mathbf q} )\cong} \mathbf{N}(\mathbf{q}_h,{\dot{\mathbf q}}_h)=\boldsymbol\tau_h-{\hat{\mathbf M}\ddot{\mathbf q}}_h,
\end{equation}
where $h>0$ is a fixed small delay time. Substituting (\ref{input}) and (\ref{approx}) in (\ref{sys}), the system dynamics is converted into an input as well as state delayed dynamics as
\begin{equation}\label{sys new}
{\hat{\mathbf M}\ddot{\mathbf q}+\bar{\mathbf N}({\mathbf q},\dot{\mathbf q},\ddot{\mathbf q},\ddot{\mathbf q}}_h)=\boldsymbol\tau_h,
\end{equation}
where ${\bar{\mathbf N}=(\mathbf M-\hat{\mathbf M})\ddot{\mathbf q}+\hat{\mathbf M}\ddot{\mathbf q}}_h-{\hat{\mathbf M}\mathbf u+ \mathbf N}$. Let $\mathbf{q}^d(t)$ be the desired trajectory to be tracked and $\mathbf{e}_1(t)=\mathbf{q}(t)-\mathbf{q}^d(t)$ is the tracking error. The auxiliary control input $\mathbf{u}$ is defined in the following way
\begin{equation}\label{aux}
{\mathbf u}(t)={\ddot{\mathbf q}}^d(t)-{\mathbf K_2\dot{\mathbf e}_1}(t)-{\mathbf K_1 \mathbf e_1}(t),
\end{equation}
where ${\mathbf K_1}$ and ${\mathbf K_2}$ are two positive definite matrices with appropriate dimensions. Putting (\ref{aux}) and (\ref{input}) in (\ref{sys new}), the following error dynamics is obtained
\begin{equation}
{\dot{\mathbf e}=\mathbf A_1 \mathbf e+\mathbf B_1 \mathbf e}_h+\mathbf{B} \boldsymbol\sigma_1,\label{new err dyn 1}
\end{equation}
where $\mathbf{e}=\begin{bmatrix}
{\mathbf e_1}\\
{\dot{\mathbf e}_1}
\end{bmatrix}$, ${\mathbf A_1}=\begin{bmatrix}
\mathbf{0} & \mathbf{I} \\
\mathbf{0} & \mathbf{0}
\end{bmatrix} 
, 
{\mathbf B_1}=\begin{bmatrix}
\mathbf{0} & \mathbf{0}\\
-{\mathbf K_1} & -{\mathbf K_2}
\end{bmatrix}$, $\mathbf{B}=\begin{bmatrix}
\mathbf{0}\\
\mathbf{I}
\end{bmatrix}$ and $\boldsymbol\sigma_1={\hat{\mathbf M}}^{-1}{(\hat{\mathbf N}}_h-{\bar{\mathbf N})+\ddot{\mathbf q}}^d_h-{\ddot{\mathbf q}}^d$ is treated as the overall uncertainty or TDE error. 

{{Herein,} the term \textit{uncertainty} denotes {perturbation
		due to parametric variations and bounded external disturbance torque} while considering that the external disturbance does not affect {the observability of the original system} \cite{Ref:6,Ref:6-2}}. {Note that 
	\[\mathbf e_h=\mathbf{e}(t-h)=\mathbf{e}(t)-\int_{-h}^0 {\dot{\mathbf e}}(t+\theta)\mathrm{d}\theta,\] 
	where} the derivative inside the integral is with respect to $\theta$, (\ref{new err dyn 1}) can be further modified as
\begin{align}
{\dot{\mathbf e}}(t)&=\mathbf{Ae}(t)-{\mathbf B_1}\int_{-h}^0 {\dot{\mathbf e}}(t+\theta)\mathrm{d}\theta+\mathbf{B}\boldsymbol\sigma_1,\label{error dyn delayed}
\end{align}
where ${\mathbf A = \mathbf A_1 + \mathbf B_1}$. It is assumed that the choice of controller gains ${\mathbf K_1}$ and ${\mathbf K_2}$ makes the matrix $\mathbf{A}$ Hurwitz which is always possible. 

{It is to be noted that the original system~\eqref{sys} is delay-free. However in TDC, the time delay $h$ in (\ref{approx}) is {artificially introduced on purpose} 
	to approximate the term $\mathbf N$ using the past time-delayed input and state information, which indeed reduces the modeling effort.} 
\begin{assumption}\label{assum 2}
	The desired trajectories are selected in a way such that $\mathbf{q}^d, {\dot{\mathbf q}}^d, {\ddot{\mathbf q}}^d \in \mathcal{L}_{\infty}$.
\end{assumption}

\begin{lemma}
	\label{lemma}
	TDE error $\boldsymbol \sigma_1$ remains bounded for the system~(\ref{sys}) if $\hat{\mathbf M}$ is selected in a way such that the following condition holds \cite{Ref:6,Ref:6-2,Ref:33}:
	\begin{equation}
	\Vert\mathbf{M}^{-1}{(\mathbf q)\hat{\mathbf M}- \mathbf I}\Vert <1. \label{mass}
	\end{equation}
\end{lemma}

Since $\boldsymbol{\sigma}_1$ remains bounded when (\ref{mass}) is satisfied from Lemma \ref{lemma}, $\exists c \in \mathbb{R}^{+}$ such that $\Vert\boldsymbol{\sigma}_1\Vert \leq c$. The term $c$ is considered only for analytical purpose and it is not utilized to design controller in this paper.

A new stability criterion of TDC, based on the Lyapunov-Krasvoskii method, is presented through Theorem 1 which presents a selection criterion and relation between the controller gains ${\mathbf K_1},{\mathbf K_2}$ and delay $h$. 

\begin{theorem}
	The system (\ref{sys new}) employing the control input (\ref{input}), having auxiliary control input (\ref{aux}) is Uniformly Ultimately Bounded (UUB) if the selection of controller gains and delay time satisfy the following condition:
	\begin{equation}\label{delay value}
	\mathbf{\Psi}=\begin{bmatrix}
	{\mathbf Q- \mathbf E}-(1+\xi)\frac{h^2}{\beta}\mathbf{D} & \mathbf{0}\\ 
	\mathbf{0} & (\xi-1)\frac{h^2}{\beta}\mathbf{D}
	\end{bmatrix}>\mathbf{0},
	\end{equation}
	
	\noindent where ${\mathbf E= \beta \mathbf P \mathbf B_1( \mathbf A_1 \mathbf D}^{-1}{\mathbf A}_1^T+{\mathbf B_1 \mathbf D}^{-1}{\mathbf B}_1^{T}+\mathbf{D}^{-1} ){\mathbf B}_1^{T}\mathbf{P}$, $\xi>1$ and $\beta>0$ are scalars, and $\mathbf{P}>\mathbf{0}$ is the solution of the Lyapunov equation $\mathbf{A}^T\mathbf{P}+{\mathbf P \mathbf A = - \mathbf Q}$ for some $\mathbf{Q}>\mathbf{0}$. 
\end{theorem}

\begin{proof}
	Let us consider the following Lyapunov function:
	\begin{equation}\label{lyapunov}
	V=V_1(\mathbf{e})+V_2(\mathbf{e})+V_3(\mathbf{e})+V_4(\mathbf{e}),
	\end{equation} 
	\noindent where
	\begin{align*}
	V_1(\mathbf{e})&={\mathbf e^T \mathbf P \mathbf e},\\ V_2(\mathbf{e})&=\frac{h}{\beta}\int_{-h}^{0}\int_{t+\theta}^{t}\mathbf{e}^T(\psi )\mathbf{De}(\psi )\mathrm{d}\psi \mathrm{d}\theta,\\
	V_3(\mathbf{e})&=\frac{h}{\beta}\int_{-h}^{0}\int_{t+\theta}^{t}\mathbf{e}^T(\psi-h )\mathbf{De}(\psi-h )\mathrm{d}\psi \mathrm{d}\theta,\\
	V_4(\mathbf{e})&=\xi \frac{h^2}{\beta}\int_{t-h}^{t}\mathbf{e}^T(\psi)\mathbf{De}(\psi)\mathrm{d}\psi.
	\end{align*}
	Using (\ref{error dyn delayed}), the time derivative  of $V_1(\mathbf{e})$ yields
	\begin{align}\label{lya_dot for time delay}
	\dot{V}_1(\mathbf{e})&=-{\mathbf e^T \mathbf Q \mathbf e} - 2\mathbf{e}^{T}{\mathbf P \mathbf B_1}\int_{-h}^{0}[ {\mathbf A_1 \mathbf e}(t+\theta ) \nonumber\\
	& +{\mathbf B_1 \mathbf e}(t-h+\theta )+\mathbf{B}\boldsymbol\sigma _1(t+\theta ) ]\mathrm{d}\theta   +2 {\hat{\mathbf s}}^T \boldsymbol\sigma_1 ,
	\end{align}
	where ${\hat{\mathbf s}=\mathbf B}^T\mathbf{Pe}$. 
	For any two non-zero vectors ${\mathbf z_1}$ and ${\mathbf z_2}$, there exists a scalar $\beta>0$ and matrix $\mathbf{D}>\mathbf{0}$ such that 
	\begin{equation}\label{ineq 2}
	\pm 2 {\mathbf z}_1^T {\mathbf z_2}\leq \beta {\mathbf z}_1^{T}\mathbf{D}^{-1}{\mathbf z_1}+(1/\beta ){\mathbf z}_2^{T}{\mathbf D \mathbf z_2}. 
	\end{equation}
	Using Jensen's inequality, the following inequality holds \cite{Ref:31}:
	\begin{align}
	\int_{-h}^{0}\mathbf{e}^{T}(\psi)\mathbf{De}(\psi)\mathrm{d}\psi \geq \frac{1}{h}\int_{-h}^{0}\mathbf{e}^{T}(\psi)\mathrm{d}\psi \mathbf{D} \int_{-h}^{0}\mathbf{e}(\psi)\mathrm{d}\psi.\label{lemma 1}
	\end{align}
	Applying (\ref{ineq 2}) and (\ref{lemma 1}), the followings are obtained:
	\begin{align}
	&- 2\mathbf{e}^{T}\mathbf{PB}_1 {\mathbf A_1}\int_{-h}^{0} \mathbf{e}(t+\theta)\mathrm{d}\theta \leq \beta \mathbf{e}^{T}\mathbf{PB}_1 {\mathbf A_1 \mathbf D}^{-1}{\mathbf A}_1^{T}\nonumber\\
	& \qquad \qquad \times {\mathbf B}_1^{T}\mathbf{Pe} +\frac{h}{\beta}\int_{-h}^{0}\mathbf{e}^{T}(t+\theta )\mathbf{ D e}(t+\theta ) \mathrm{d}\theta, \label{cond1}\\
	&- 2\mathbf{e}^{T}{\mathbf P \mathbf B_1 \mathbf B_1}\int_{-h}^{0}  \mathbf{ e}(t-h+\theta)\mathrm{d}\theta \leq \beta \mathbf{e}^{T}{\mathbf P \mathbf B_1 \mathbf B_1 \mathbf D}^{-1}   \nonumber\\
	&\times \mathbf{B}_1^{T}\mathbf{B}_1^{T}\mathbf{Pe}+\frac{h}{\beta}\int_{-h}^{0}\mathbf{e}^{T}(t-h+\theta )\mathbf{D e}(t-h+\theta )\mathrm{d}\theta , \label{cond2} \\
	&- 2\mathbf{e}^{T}\mathbf{PB}_1\int_{-h}^{0} [ \mathbf{B} \boldsymbol \sigma _1(t+\theta)]\mathrm{d}\theta \leq \beta \mathbf{e}^{T}{\mathbf P \mathbf B_1 \mathbf D}^{-1}\mathbf{B}_1^{T}\mathbf{Pe}   \nonumber\\ 
	& \qquad \quad  + \frac{h}{\beta}\int_{-h}^{0}(\mathbf{B} \boldsymbol \sigma _1(t+\theta))^{T} \mathbf{D B} \boldsymbol \sigma _1(t+\theta)  \mathrm{d}\theta. \label{cond3}
	\end{align}
	Assuming that system remains locally Lipschitz within the delay, then $ \exists \Gamma_1>0$ such that the following holds:
	\begin{align}\label{cond41}
	\frac{h}{\beta }\left \| \int_{-h}^{0}\left [ (\mathbf{B} \boldsymbol\sigma _1(t+\theta))^{T}\mathbf{ D B} \boldsymbol \sigma _1(t+\theta) \right ]\mathrm{d}\theta  \right \|\leq \Gamma_1. 
	\end{align}

	
	\noindent Again,	
	\begin{align}
	\dot{V}_2(\mathbf{e})&=\frac{h^2}{\beta}\mathbf{e}^T\mathbf{De}-\frac{h}{\beta}\int_{-h}^{0}\mathbf{e}^T(t+\theta) \mathbf{D e}(t+\theta)\mathrm{d}\theta, \label{dot v2 new}\\ 
	\dot{V}_3(\mathbf{e}) &=\frac{h^2}{\beta}\mathbf{e}^T_h \mathbf{D e}_h-\frac{h}{\beta}\int_{-h}^{0}\mathbf{e}^T(t-h+\theta) \mathbf{D e}(t-h+\theta)\mathrm{d}\theta, \label{dot v3 new}\\  
	\dot{V}_4(\mathbf{e})&=\xi \frac{h^2}{\beta}(\mathbf{e}^T{\mathbf D \mathbf e}-\mathbf{e}^T_h\mathbf{De}_h).\label{dot v4 new}
	\end{align}
	For a positive scalar $\iota=\Vert\mathbf{B}^T\mathbf{P}\Vert$, we have $\Vert {\hat{\mathbf s}} \Vert \leq \iota  \Vert \mathbf{\bar{e}}\Vert$, where ${\hat{\mathbf s}}=\begin{bmatrix}
	\mathbf{B}^T\mathbf{P} & \mathbf 0
	\end{bmatrix}\mathbf{\bar{e}}$ and $\mathbf{\bar{e}}=\begin{bmatrix}
	\mathbf{e}^T & \mathbf{e}^T_h
	\end{bmatrix}^T$. Let controller gains ${\mathbf K_1, \mathbf K_2}$ and delay time $h$ are selected to make $\mathbf{\Psi>0}$. Substituting (\ref{cond1})-(\ref{cond41}) into (\ref{lya_dot for time delay}) and using (\ref{dot v2 new})-(\ref{dot v4 new}) yield
	\begin{align}
	\dot{V}& \leq -{\bar{\mathbf e}}^T {\boldsymbol\Psi \bar{\mathbf e}}+\Gamma_1+2{\hat{\mathbf s}}^T \boldsymbol\sigma_1 \nonumber \\
	& \leq -z \Vert \bar{\mathbf e} \Vert^2-(\lambda_{\min}(\mathbf\Psi)-z)\Vert \bar{\mathbf e}\Vert^2+\Gamma_1+2\iota c \Vert \mathbf{\bar{e}}\Vert. 
	\end{align}
	where $0<\bar{z}<\lambda_{\min}(\mathbf\Psi)$. Then, $\dot{V}<0$ would be established if $(\lambda_{\min}(\mathbf\Psi)-z)\Vert {\bar{\mathbf e}}\Vert^{2} >  \Gamma_1 +2 \iota c   \Vert {\bar{\mathbf e}}\Vert$.
	The system (\ref{sys new}) is thus UUB \cite{Ref:24}. 
%
\end{proof}
%

\subsection{Time-Delayed Adaptive Robust Control (TARC)}
\label{sec: 2.2}
The performance of TDC gets affected by the presence of $\boldsymbol \sigma_1$, due to the absence of any robustness term. 
Note that 
in practical circumstances, it is not always possible to determine either the complete model of the system \cite{Ref:21,Ref:22}, or predefined uncertainty bound \cite{Ref:2,Ref:13,Ref:16,Ref:17,Ref:18,Ref:19} due to unmodelled dynamics. Further, the ARCs designed in \cite{Ref:21,Ref:22} suffer from over- and under-estimation problems. 

Moreover, it can be noticed from (\ref{approx}) and (\ref{aux}) that the state-derivatives are necessary to compute the control law of TDC. However, in many circumstances ${\dot{\mathbf q}, \ddot{\mathbf q}}$ are not available explicitly. Under such scenario, one has to approximate the state-derivatives ${\dot{\mathbf q}, \ddot{\mathbf q}}$ by Euler backward numerical derivative technique (forward numerical derivative is not possible as future data is not available) to implement TDC \cite{Ref:9,Ref:7,Ref:9_1,Ref:12,Ref:34,Ref:10}. However, to the best knowledge of the authors, effect of such numerical approximation error in system stability is yet to be studied in the literature of the TDE based controllers \cite{Ref:jin,Ref:6,Ref:6-2,Ref:26,Ref:new1,Ref:33,Ref:7,Ref:9,Ref:9_1,Ref:10,Ref:12,Ref:34}.  
Yet again, state information (or position for EL system) $\mathbf{q}$ is often contaminated with noise and numerical evaluation of ${\dot{\mathbf q}, \ddot{\mathbf q}}$ under such circumstances degrades the controller performance. 

Therefore, with consideration of the limitations in the existing controllers, 
a novel adaptive-robust control law, named Time-Delayed Adaptive Robust Control (TARC) is proposed in this endeavour. 
TARC neither requires the complete model, nor any predefined bound of the uncertainties, while it 
alleviates the over- and under-estimation 
problems of the switching gain. 
Moreover, to circumvent the measurement error, a state-derivatives estimation technique \cite{Ref:30} is incorporated in TARC which estimates the velocity and acceleration terms for the EL systems (\ref{sys}) from the position information of past time instances.

Before presenting the control structure of the proposed TARC, the following Lemma is stated which helps to estimate the state derivative terms: 
\begin{lemma} \label{lemma est}
	For time $t \geq \varsigma$, the $j$-th order time derivative of the $\Lambda$-th degree polynomial $\mathbf{q}$ in (\ref{sys new}) can be computed in the following manner\cite{Ref:30}:
	\begin{align}
	\mathbf{\hat{q}}^{(j)}(t)&=\int_{-\varsigma}^{0}\Omega _{j}(\psi)\mathbf{q}(t+\psi)\mathrm{d}\psi\label{lemma 2}\\
	\Omega _{j}(\psi) &=\frac{(\Lambda+1+j)!}{\varsigma ^{j+1}j!(\Lambda-j)!}\nonumber \\ 
	& \times \sum_{k=0}^{\Lambda}\frac{(-1)^{k}(\Lambda+1+k)!}{(j+k+1)(\Lambda-k)!(k!)^{2}} \left ( \frac{\psi}{\varsigma} \right )^{k}.\label{omega}
	%
	\end{align}
\end{lemma}
{Note that the authors in \cite{Ref:7} have applied a low pass filter (LPF) separately to mitigate the effect of measurement error arising from numerical differentiation. 
	However, as stated in~\cite{Ref:7}, inclusion of an external LPF reduces $\hat{\mathbf M}$ which has adverse impact on the controller performance as well as stability condition due to the reduced the stability region (low value of $\hat{\mathbf M}$ pushes the boundedness condition (\ref{mass}) towards the perimeter of the unit circle). Hence, the designer has to make a trade-off between noise attenuation and controller performance.
	 
	{Whereas,} the process (\ref{lemma 2}) itself has noise attenuation capability due to the integral term {as mentioned in \cite{Ref:30}}. 
	{Furthermore}, on the contrary to \cite{Ref:jin,Ref:6,Ref:6-2,Ref:26,Ref:new1,Ref:33,Ref:7,Ref:9,Ref:9_1,Ref:10,Ref:12,Ref:34}, the closed loop stability of TARC is {explicitly carried out in Section \ref{sec:stable} by considering (\ref{lemma 2})}. The stability analysis shows that inclusion of the process (\ref{lemma 2}) does not have any impact on $\hat{\mathbf M}$. Hence, the state-derivatives approximation technique (\ref{lemma 2}) does not require any separate LPFs and thus can avoid any pervasive effects on system stability. } 

The structure of the control input of TARC is similar to (\ref{input}) except the auxiliary control input $\mathbf{u}$ and ${\hat{\mathbf N}}$  selected as follows:
\begin{align}
\mathbf{u} &= \hat{\mathbf u}+\Delta \mathbf u, \label{tarc input} \\
\hat{\mathbf{u}} &={\ddot{\mathbf q}}^d - {\mathbf K_1 \mathbf e_1} - {\mathbf K_2\dot{\hat{\mathbf e}}}_1 \label{u hat},\\
{\hat{\mathbf N}}& \cong \mathbf{N}_{h}=\boldsymbol\tau_{h}-{\hat{\mathbf M}\ddot{\hat{\mathbf q}}}_{h}, \label{h hat new}
\end{align}
where ${\dot{\hat{\mathbf e}}}_1=\dot{\hat{\mathbf q}}-\dot{\mathbf q}^d$, and  ${\dot{\hat{\mathbf q}}}$ and ${\ddot{\hat{\mathbf q}}}$ are evaluated from (\ref{lemma 2}) and (\ref{omega}). $\mathbf{\hat{u}}$ is the nominal control input; $\Delta \mathbf{u}$ is a switching control law which acts as a robustness term to negotiate 
the TDE error, defined as follows:
\begin{align}
\Delta \mathbf{u}=
\begin{cases}
-\alpha\hat{c} ~{\mathbf{s}}/{\Vert \mathbf{s} \Vert}, & ~~~ \text{if } \Vert \mathbf{s} \Vert\geq \epsilon,\\
-\alpha\hat{c}~{\mathbf{s}}/{\epsilon},        & ~~~ \text{if } \Vert \mathbf{s} \Vert< \epsilon,\\
\end{cases}\label{delta u}
\end{align}
where $\mathbf{s=B}^T\mathbf{P}[{\mathbf e}_{1}^T \quad \mathbf 0]^T$; and $\alpha \geq 1$ is a scalar adaptive gain; $\gamma>0$ and $\epsilon>0$ represent two small scalars.
In this paper, a novel adaptive control law to compute $\hat{c}$ is proposed as follows:
\begin{align}
\dot{\hat{c}}& =
\begin{cases}
~~\bar{c} \Vert\mathbf{s}\Vert,      &~{\hat{c}} \leq \gamma \lor f(\mathbf{e}_1)>0\\
-\underline{c}\Vert \mathbf{s} \Vert,              &~f(\mathbf{e}_1) \leq 0 \\
\end{cases} \label{ATRC}
\end{align}

{\noindent with $\hat{c}(t_0)  > \gamma$,} where $t_0$ is the initial time; $\bar{c}>0, \underline{c} >0$ are two user defined scalars; $f(\mathbf{e}_1) \in \mathbb{R}^n \mapsto \mathbb{R}^{+}$ is a suitable function of the error defined by the designer and it is to be selected in a way such that $f(\mathbf{e}_1)>0$ (resp. $f(\mathbf{e}_1) \leq 0$) defines the instances when tracking error increases (resp. does not increase); here, 
it is selected as $ f(\mathbf{e}_1)=\Vert \mathbf{s} \Vert - \Vert \mathbf{s}_h  \Vert$. 

According to the adaptive law (\ref{ATRC}) and the chosen 
$f(\mathbf{e})$, $\hat{c}$ increases (resp. decreases) whenever error trajectories move away (resp. do not move away) from $|| \mathbf{s} ||=0$. 
Let us define 
{
	\[{\hat{\mathbf s}}=\mathbf{s}+ \Delta \mathbf{s},\] where
	$\Delta \mathbf{s=B}^T\mathbf{P}
	\begin{bmatrix}
	\mathbf{0}\\ {\dot{\mathbf e}_1}
	\end{bmatrix}$, 
	and 
	$\mathbf{P}=\begin{bmatrix}
	{\mathbf P_1} & {\mathbf P}_2^T\\ 
	{\mathbf P_2} & {\mathbf P_3}
	\end{bmatrix}.$}


\noindent By 
evaluating the structure of ${\hat{\mathbf s}}$, the following relation is established:
\begin{align*}
{\hat{\mathbf s}}&=\boldsymbol{\iota}_1 \mathbf e_f, \text{where} ~~ \boldsymbol{\iota}_1=\begin{bmatrix}
{\mathbf P_2} & {\mathbf P_3} & \mathbf{0} & \mathbf{0} & \mathbf{0} & \mathbf{0} & \mathbf{0} & \mathbf{0} 
\end{bmatrix},\\
{\mathbf e_f}&= {\left[
	\mathbf{e}^T ~ \mathbf{e}^T_h ~ \int_{-\varsigma}^{0}\Omega _{1}\mathbf{e}^T(t+\psi) \mathrm{d}\psi ~ \int_{-\varsigma}^{0}\Omega _{1}\mathbf{e}^T_h(t+\psi) \mathrm{d}\psi \right]}^T.
\end{align*}
The controller gains ${\mathbf K_1}$ and ${\mathbf K_2}$, design parameter $\alpha$, and the matrix $\mathbf{Q}$ are selected in a way such that the following hold:
\begin{align}
\mathbf P_3 \mathbf P_2^T > \mathbf 0 \label{c2}.
\end{align}

{\noindent The relation (\ref{c2}) is used for the stability analysis in Section~\ref{sec:stable}.}


{\subsection{Comparison with Existing ARC} 
	Compared with the ASMC developed in~\cite{Ref:21,Ref:22}, the proposed TARC prevents the over- and under-estimation problems of the switching gain. To elaborate, the adaptive law of~\cite{Ref:21,Ref:22} is presented as follows:
	\begin{equation}
	\dot{\varrho}=\begin{cases}
	\bar{\varrho}\text{sgn}( \Vert \mathbf{r} \Vert - \delta)     & , \text{if}~ \varrho>\bar{\gamma},\\
	\qquad \bar{\gamma}          & , \text{if}~ \varrho \leq \bar{\gamma},\\
	\end{cases}\label{Asmc}
	\end{equation}
	where $\varrho$ denotes the switching gain, $\mathbf{r}$ denotes a sliding surface, and 
	$\bar{\varrho}, \delta, \bar{\gamma} \in \mathbb{R}^{+}$ are user-defined scalars. 
	It can be observed from (\ref{Asmc}) that when $\Vert\mathbf r\Vert \geq \delta$ (resp. $\Vert\mathbf r\Vert < \delta$), the switching gain $\varrho$ increases (resp. decreases) monotonically, even if the error trajectories move close to (resp. away from) $\Vert\mathbf r\Vert = 0$. It thus gives rise to the potential overestimation (resp. underestimation) problem of the switching gain. Further, very low (resp. high) value of $\delta$ may force $\varrho$ to increase (resp. decrease) for longer duration when $\Vert\mathbf{r}\Vert \geq \delta$ (resp. $\Vert\mathbf{r}\Vert < \delta$). This in turn may escalate the overestimation (resp. underestimation) problem. Hence, a designer needs to exhaustively tune the \textit{predefined fixed threshold value} $\delta>0$ to tackle the over- and under-estimation problems in the adaptive laws of ASMC in (\ref{Asmc}). This tuning procedure will vary from system to system and it is even more difficult for multiple DoFs system under uncertain operational/working scenarios.
	
	In contrast, the proposed adaptive law of TARC, shown in (\ref{ATRC}), does not involve any threshold value. The switching gain $\hat{c}$ increases (resp. decreases) when the error trajectories move away (resp. do not move away) from  $\lVert\boldsymbol s \rVert$$=$$0$. This in turn permits TARC to alleviate the over- and under-estimation problems. The aforementioned increasing-decreasing nature of $\hat{c}$ certainly avoids making $\hat{c}$ a monotonically increasing function like \cite{Ref:20,Ref:20-1,Ref:20-2,Ref:20-3,Ref:32}.}



\section{Stability and Parameter Selection of TARC} 
\subsection{Stability Analysis of TARC}\label{sec:stable}
In this section, the stability of the system (\ref{sys new}) employing TARC is analysed in the sense of UUB. Before formally stating the stability result using TARC through Theorem \ref{th tarc}, the following Lemma is defined: 
\begin{lemma}\label{lemma double jensen}
	For any non zero vector $\boldsymbol{\vartheta}(\theta, \psi)$, constant matrix $\mathbf{F>0}$ the following relation holds:
	\begin{align}
	&\int_{-h}^{0}\int_{-\varsigma}^{0}\boldsymbol{\vartheta} ^T(\theta, \psi) \mathbf{F} \boldsymbol{ \vartheta}(\theta, \psi) \mathrm{d}\psi \mathrm{d}\theta \geq  \nonumber\\ 
	&\frac{1}{h \varsigma}\left(  \int_{-h}^{0}\int_{-\varsigma}^{0}\boldsymbol{\vartheta} ^T(\theta, \psi)\mathrm{d}\psi \mathrm{d}\theta\right) \mathbf{F} \left( \int_{-h}^{0}\int_{-\varsigma}^{0}\boldsymbol{\vartheta}(\theta,\psi)\mathrm{d}\psi \mathrm{d}\theta\right). \label{lemma 3}
	\end{align}
\end{lemma}

%

\begin{theorem} \label{th tarc}
	The closed-loop system (\ref{sys new}) employing (\ref{input}) and (\ref{tarc input}) and having the adaptive law (\ref{ATRC}) is UUB, provided the selections of ${\mathbf K_1, \mathbf K_2}, h$ and $\varsigma$ satisfy the following condition:
	\begin{align}
	&\begin{bmatrix}
	\mathbf{Q-\bar{E}}-(1+\xi)\frac{h^2}{\beta}\mathbf{D}& {\mathbf P \breve{\mathbf B}} & {\mathbf P\bar{\mathbf B}\mathbf J}\\ 
	{\breve{\mathbf B}^T \mathbf P} & (\xi-1)\frac{h^2}{\beta}{\mathbf D-\bar{\mathbf F}} & \mathbf{0}\\
	{\mathbf J^T \bar{\mathbf B}^T}\mathbf{P} & \mathbf{0} & \mathbf{J}^T\mathbf{LJ}
	\end{bmatrix} \nonumber \\
	& \qquad \qquad \qquad \qquad \qquad =\mathbf{\Theta>0} \label{delay value 1},
	\end{align} 
	\noindent where $\xi>1$, $\beta>0$, $\bar{\mathbf E}=\beta \mathbf P {\mathbf B}_1( {\mathbf A}_1 \mathbf D^{-1}{\mathbf A}_1^T+{\mathbf B}_1 \mathbf D^{-1}{{\mathbf B}}_1^{T}+{\mathbf D}^{-1}+ \bar{\mathbf B} \mathbf D^{-1} {\bar{\mathbf B}}^{T} )\mathbf{{B}}_1^{T}\mathbf{P}$, ${\bar{\mathbf F}}=(\frac{h^2}{\beta}\mathbf{D+L}) \varsigma \int_{-\varsigma}^{0}\Omega _{1}^2(\psi)\mathrm{d}\psi$, $\mathbf{L>0}$, ${\bar{\mathbf B}= \mathbf B}\begin{bmatrix}
	{\mathbf K_2} & \mathbf{0}
	\end{bmatrix}$, ${\breve{\mathbf B}=\mathbf B}\begin{bmatrix}
	\mathbf{0} & {\mathbf K_2}
	\end{bmatrix}$, and $\mathbf{J}=\begin{bmatrix}
	\mathbf{0} & \mathbf{I}
	\end{bmatrix}$.
\end{theorem}

\begin{proof} 
	Let us define the Lyapunov functional as
	\begin{align}
	V_r&=V_f(\mathbf{e})+\frac{1}{\bar{c}}(\hat{c}-c)^2,\label{tarc lya}\\
	\text{where}~V_f(\mathbf{e})=&V(\mathbf{e})+V_{f1}(\mathbf{e})+V_{f2}(\mathbf{e})+V_{f3}(\mathbf{e}),\nonumber \\
	V_{f1}(\mathbf{e})=&\frac{h \varsigma}{\beta}\int_{-h}^{0}\int_{-\varsigma}^{0}\int_{t-h+\psi}^{t-h}\mathbf{e}^T(\eta +\theta)\mathbf{D} \nonumber\\
	& \times \Omega _{1}^2 (\psi)\mathbf{e}(\eta+\theta)\mathrm{d}\eta \mathrm{d}\psi \mathrm{d}\theta, \nonumber\\
	V_{f2}(\mathbf{e})=&\frac{h \varsigma}{\beta}\int_{-h}^{0}\int_{-\varsigma}^{0}\int_{t+\theta}^{t}\mathbf{e}^T(\eta-h)\mathbf{D}  \nonumber\\
	& \times \Omega _{1}^2(\psi)\mathbf{e}(\eta-h)\mathrm{d}\eta \mathrm{d}\psi \mathrm{d}\theta, \nonumber\\
	V_{f3}(\mathbf{e})=&\varsigma \int_{-\varsigma}^{0}\int_{t+\psi}^{t}\mathbf{e}^T(\eta-h)\mathbf{L} \Omega _{1}^2(\psi)\mathbf{e}(\eta-h)\mathrm{d}\eta \mathrm{d}\psi. \nonumber
	\end{align}
	\noindent  $V$ is given in (\ref{lyapunov}).
	Again, substituting 
	(\ref{tarc input}) into (\ref{sys new}), the error dynamics becomes
	\begin{equation}\label{error dyn state new}
	{\dot{\mathbf e}={\mathbf A}_1 \mathbf e+{\mathbf B}_1 \mathbf e}_h-{\bar{\mathbf B}}\int_{-\varsigma}^{0}\Omega _{1}(\psi )\mathbf{e}(t-h+\psi )\mathrm{d}\psi+\mathbf{B} \boldsymbol\sigma,
	\end{equation}
	where $\boldsymbol\sigma=\Delta \mathbf{u}_h+ {\mathbf K_2\dot{\mathbf e}}_{1h}+\boldsymbol \sigma_1$. 
	Further, the error dynamics (\ref{error dyn state new}) can be written as  
	\begin{align}
	{\dot{\mathbf e}}=&{\mathbf A \mathbf e-{\mathbf B}_1} \int_{-h}^{0} {\dot{\mathbf e}}(t+\theta)\mathrm{d} \theta-{\bar{\mathbf B}}\int_{-\varsigma}^{0}\Omega _{1}(\psi )\mathbf{e}(t-h+\psi )\mathrm{d}\psi \nonumber\\
	& + \mathbf{B}\boldsymbol\sigma.\label{error dyn state new 2}
	\end{align}	
	\noindent For ease of analysis, we have segregated the stability analysis into two parts: (1) derivation of $\dot{V}_f$ and (2) derivation of $\dot{V}_r$.
		{ \textbf{(1) Derivation of $\dot{V}_f$}:
		Using (\ref{error dyn state new 2}), the time derivative  of $V_1(\mathbf{e})$ yields
		\begin{align}\label{lya_dot for time delay new}
		\dot{V}_1(\mathbf{e})=&-{\mathbf e^T \mathbf Q \mathbf e} - 2\mathbf{e}^{T}{\mathbf P \mathbf B_1}\int_{-h}^{0}\dot{\mathbf e}(t+\theta )\mathrm{d}\theta  \nonumber\\
		& - 2\mathbf{e}^{T}{\mathbf P \bar{\mathbf B}} \int_{-\varsigma}^{0} \Omega _{1}(\psi )\mathbf{e}(t-h+\psi )\mathrm{d}\psi +2 {\hat{\mathbf s}}^T \boldsymbol\sigma.
		\end{align}
		Using (\ref{error dyn state new}), the second term of (\ref{lya_dot for time delay new}) can be expanded as
		\begin{align}
		&- 2\mathbf{e}^{T}{\mathbf P \mathbf B_1}\int_{-h}^{0}\dot{\mathbf e}(t+\theta )\mathrm{d}\theta =- 2\mathbf{e}^{T}{\mathbf P \mathbf B_1}\int_{-h}^{0} [{\mathbf A_1 \mathbf e}(t+\theta ) \nonumber\\
		& \qquad \qquad +{\mathbf B_1 \mathbf e}(t-h+\theta ) +\mathbf{B} \boldsymbol\sigma (t+\theta) \nonumber\\
		&\qquad \qquad -{\bar{\mathbf B}}\int_{-\varsigma}^{0}\Omega _{1}(\psi )\mathbf{e}(t-h+\theta+\psi )\mathrm{d}\psi]\mathrm{d} \theta. \label{ex 1}
		\end{align}
		The first two terms of (\ref{ex 1}) can be represented exactly like (\ref{cond1}) and  (\ref{cond2}). Further, following  (\ref{cond3}) we have
		\begin{align}
		&- 2\mathbf{e}^{T}\mathbf{PB}_1\int_{-h}^{0} [ \mathbf{B} \boldsymbol \sigma (t+\theta)]\mathrm{d}\theta \leq \beta \mathbf{e}^{T}{\mathbf P \mathbf B_1 \mathbf D}^{-1}\mathbf{B}_1^{T}\mathbf{Pe}   +\Gamma, \label{cond3 new}
		\end{align}
		where $\Gamma \geq \frac{h}{\beta }\left \| \int_{-h}^{0}\left [ (\mathbf{B} \boldsymbol \sigma(t+\theta))^{T} \mathbf{ D B} \boldsymbol \sigma(t+\theta) \right ]\mathrm{d}\theta  \right \|$. Applying (\ref{ineq 2}), the last term of (\ref{ex 1}) can be represented as the following:
		\begin{align}
		& 2\mathbf{e}^{T}{\mathbf P \mathbf B_1} \bar{\mathbf B} \int_{-h}^{0} \int_{-\varsigma}^{0}\Omega _{1}(\psi )\mathbf{e}(t-h+\theta+\psi )\mathrm{d}\psi \mathrm{d} \theta \leq \nonumber \\
		& \qquad \beta \mathbf{e}^{T}\mathbf{PB}_1 \bar{\mathbf B} {\mathbf D}^{-1} \bar{\mathbf B}^{T} {\mathbf B}_1^{T}\mathbf{Pe} \nonumber\\
		&\qquad +\frac{1}{\beta} \int_{-h}^{0} \int_{-\varsigma}^{0}\Omega _{1}(\psi )\mathbf{e}^T(t-h+\theta+\psi ) \mathrm{d}\psi \mathrm{d}\theta \mathbf D \nonumber \\
		& \qquad \qquad \times \int_{-h}^{0} \int_{-\varsigma}^{0}\Omega _{1}(\psi )\mathbf{e}(t-h+\theta+\psi ) \mathrm{d}\psi \mathrm{d}\theta. \label{ex 2}
		\end{align}
		Applying (\ref{lemma 3}) to the last term of (\ref{ex 2}) yields
		\begin{align}
		& 2\mathbf{e}^{T}{\mathbf P \mathbf B_1} \bar{\mathbf B} \int_{-h}^{0} \int_{-\varsigma}^{0}\Omega _{1}(\psi )\mathbf{e}(t-h+\theta+\psi )\mathrm{d}\psi \mathrm{d} \theta \leq \nonumber \\
		& \qquad \beta \mathbf{e}^{T}\mathbf{PB}_1 \bar{\mathbf B} {\mathbf D}^{-1} \bar{\mathbf B}^{T} {\mathbf B}_1^{T}\mathbf{Pe} \nonumber\\
		&\qquad +\frac{h \varsigma}{\beta} \int_{-h}^{0} \int_{-\varsigma}^{0} [\Omega _{1}^2(\psi )\mathbf{e}^T(t-h+\theta+\psi ) \mathbf D \nonumber \\
		& \qquad \qquad \times \mathbf{e}(t-h+\theta+\psi )] \mathrm{d}\psi \mathrm{d}\theta. \label{ex 3}
		\end{align}
		Substituting (\ref{cond1}), (\ref{cond2}), (\ref{cond3 new}) and (\ref{ex 3}) into (\ref{ex 1}) and, then using $\dot{V}_2,\dot{V}_3,\dot{V}_4$ from (\ref{dot v2 new}), (\ref{dot v3 new}) and (\ref{dot v4 new}) respectively, we have
		\begin{align}
		&\dot{V}(\mathbf{e})\leq -\mathbf{e}^{T}\left [{ \mathbf Q-\bar{ \mathbf E}} -(1+\xi)\frac{h^2}{\beta}\mathbf{D} \right]\mathbf{e}+ \Gamma+\frac{h \varsigma}{\beta} \times \nonumber\\
		&\int_{-h}^{0}\int_{-\varsigma}^{0} \mathbf{e}^T(t-h+\theta+\psi)\mathbf{D} \Omega _{1}^2(\psi)\mathbf{e}(t-h+\theta+\psi)\mathrm{d}\psi \mathrm{d}\theta \nonumber\\
		& \qquad -2\mathbf{e}^T{\mathbf P\bar{\mathbf B}}\int_{-\varsigma}^{0}\Omega _{1}(\psi)\mathbf{e}(t-h+\psi)\mathrm{d}\psi+ 2{\hat{\mathbf s}}^{T} \boldsymbol\sigma. \label{v dot tarc}
		\end{align}
	}Further, the time derivatives of $V_{f1}, V_{f2}$ and $V_{f3}$ yields
	\begin{align}
	&\dot{V}_{f1} = \frac{h \varsigma}{\beta}\int_{-h}^{0} \int_{-\varsigma}^{0}\mathbf{e}^T(t-h+\theta)\mathbf{D}\Omega _{1}^2(\psi)\mathbf{e}(t-h+\theta)\mathrm{d}\psi \mathrm{d}\theta\nonumber\\
	& \qquad -\frac{h \varsigma}{\beta}\int_{-h}^{0} \int_{-\varsigma}^{0}\mathbf{e}^T(t-h+\theta+\psi)\mathbf{D}\Omega _{1}^2(\psi)  \nonumber\\
	& \qquad \qquad \qquad \qquad \qquad \times \mathbf{e}(t-h+\theta+\psi)\mathrm{d}\psi \mathrm{d}\theta\label{v1 dot ftdc},\\
	&\dot{V}_{f2}  = \frac{h^2 \varsigma}{\beta}\mathbf{ e}^T(t-h)\mathbf{D} \int_{-\varsigma}^{0}\Omega _{1}^2(\psi)\mathrm{d}\psi \mathbf{e}(t-h)-\frac{h \varsigma}{\beta}  \nonumber \\
	& \times \int_{-h}^{0} \int_{-\varsigma}^{0} \mathbf{e}^T(t-h+\theta)\mathbf{D}\Omega _{1}^2(\psi)\mathbf{e}(t-h+\theta)\mathrm{d}\psi \mathrm{d}\theta\label{v2 dot ftdc},\\
	&\dot{V}_{f3} = \varsigma \mathbf{e}^T(t-h)\int_{-\varsigma}^{0} \mathbf{L} \Omega _{1}^2(\psi)\mathrm{d}\psi \mathbf{e}(t-h)-\int_{-\varsigma}^{0}A_{d}(\psi)  \nonumber\\
	& \times \mathbf{e}^T(t-h+\psi)\mathrm{d}\psi \mathbf{L} \int_{-\varsigma}^{0}\Omega _{1}(\psi)\mathbf{e}(t-h+\psi)\mathrm{d}\psi. \label{v3 dot ftdc}
	\end{align}
	Now, taking $2{\hat{\mathbf s}}^T {\mathbf K}_2\dot{\mathbf e}_{1h}=2\mathbf{e}^T {\mathbf P \breve{\mathbf B} \mathbf e}_h$, $\int_{-\varsigma}^{0}\Omega _{1}(\psi)\mathbf{e}(t-h+\psi)\mathrm{d}\psi=\int_{-\varsigma}^{0}\Omega _{1}(\psi)\mathbf{J} \begin{bmatrix}
	\mathbf{e}^T(t+\psi) & \mathbf{e}^T(t-h+\psi)
	\end{bmatrix}^T \mathrm{d}\psi$ and combination of (\ref{v dot tarc})-(\ref{v3 dot ftdc}) yield
	\begin{equation}\label{vf dot 1}
	\dot{V}_f(\mathbf{e})\leq -{\mathbf e}_f^T {\mathbf\Theta \mathbf e_f} + \Gamma +2{\hat{\mathbf s}}^T(\Delta \mathbf{u}+\boldsymbol \sigma_1)+2{\hat{\mathbf s}}^T (\Delta \mathbf{u}_h- \Delta \mathbf{u}).
	\end{equation}
	\noindent Therefore, ${\mathbf K_1, \mathbf K_2}$, $h$, and $\varsigma$ are required to be selected in a way such that $\mathbf{\Theta>0}$.
	Since $\Delta {\mathbf{u}}$ is piecewise continuous, then $\exists \Upsilon \in \mathbb{R}^ {+}$ such that the following holds \cite{Ref:24}
	\begin{align}
	\Vert\Delta \mathbf{u}- \Delta \mathbf{u}_h\Vert \leq \Upsilon.  
	\label{new 2}
	\end{align}
	Using (\ref{new 2}), we have the following from (\ref{vf dot 1}):
	\begin{equation}\label{vf dot}
	\dot{V}_f(\mathbf{e})\leq -{\mathbf e}_f^T {\mathbf\Theta \mathbf e_f} + \Gamma +2{\hat{\mathbf s}}^T(\Delta \mathbf{u}+\boldsymbol \sigma_1)+2{\hat{\mathbf s}}^T \Upsilon. 
	\end{equation}
	\textbf{(2) Boundedness of the switching gain of TARC}: 
It can be noted from (\ref{ATRC}) that $\hat{c}$ increases when either $\hat{c}\leq \gamma$ or $f(\mathbf{e}_1)>0$ (i.e., $\Vert \mathbf{s} \Vert$ increases for this case). However, as $\hat{c}(t_0) > \gamma$ and $\hat{c}(t_0)\geq \gamma$ $\forall t \geq t_0$, $\hat{c}$ can attempt to breach its lower bound only when it is associated with a decreasing motion. As a result, to check the boundedness condition of $\hat{c}$ it is sufficient to only consider $f(\mathbf{e}_1)>0$ or $\Vert \mathbf{s} \Vert >  \Vert \mathbf{s}_h \Vert $. This implies $\exists \varpi \in \mathbb{R}^{+}$ such that $\Vert \mathbf{s} \Vert \geq \varpi$ for this case. Then from (\ref{ATRC}) one has
\begin{equation}
\hat{c} \geq \bar{c} \varpi. \label{c low}
\end{equation} 
Now consider $\Vert \mathbf{s} \Vert \geq \epsilon$. From (\ref{vf dot}) one has
\begin{align}
	\dot{V}_f&\leq -{\mathbf e}_f^T {\mathbf\Theta \mathbf e_f}+\Gamma+2{\hat{\mathbf s}}^T(-\alpha \hat{c} {\mathbf{s}}/{\Vert\mathbf{ s} \Vert}+\boldsymbol \sigma_1) +2{\hat{\mathbf s}}^T \Upsilon\nonumber\\
	&\leq -{\mathbf e}_f^T {\mathbf\Theta \mathbf e_f}+\Gamma-2\alpha \hat{c}\Vert\mathbf{s}\Vert +2(\Upsilon+c)\Vert \iota_1 \Vert \Vert{\mathbf e}_f\Vert \nonumber\\
	& \leq - \rho_m V_f -(\lambda_{\min}{(\mathbf\Theta)-z)\Vert\mathbf e_f\Vert^2}-2\alpha \hat{c}\varpi\nonumber\\
	&\quad +2(\Upsilon+c)\Vert \iota_1 \Vert \Vert{\mathbf e}_f\Vert+\Gamma
\end{align}
Thus $\dot{V}_f <0$ is established when
\begin{equation}\label{error bound appen}
	\Vert{\mathbf e_f}\Vert \geq \mu_1+\sqrt{\frac{\Gamma}{(\lambda _{\min}(\mathbf{\Theta)}-z)}+\mu_1^2}:=\varpi_0.
	\end{equation}
Though this error bound is conservative as the effect of $\hat{c}$ is ignored, it eventually helps to put a bound on $\hat{c}$. Therefore, $\dot{V}_f= -\lambda_{\min}{(\mathbf\Theta)\Vert\mathbf e_f\Vert^2}-2\alpha \hat{c}\varpi +2(\Upsilon+c)\Vert \iota_1 \Vert \Vert{\mathbf e_f}\Vert +\Gamma<0$ when 
\begin{equation}
 \alpha \varpi \hat{c} \geq (\Upsilon+c)\Vert \iota_1 \Vert \varpi_0+\Gamma/2. \label{appen bound}
\end{equation}
Let $t_{in}$ be any arbitrary initial time when $\hat{c}$ starts increasing. Then integrating both sides of (\ref{c low}) and using (\ref{appen bound}) one can find that there exist finite time $t_1$ such that 
\begin{align}
t_1 \leq \frac{(\Upsilon+c)\Vert \iota_1 \Vert \varpi_0+\Gamma/2}{\alpha \bar{c} \varpi^2}.
\end{align}
Then using comparison lemma \cite{Ref:24} one has
\begin{align}
\dot{V}_f &\leq -\rho_1 V_f \Rightarrow \dot{V}_f(t) \leq {V}_f(t_{in}+t_1)e^{-(t-t_1)}~ \forall  t \geq t_{in}+t_1, \label{v dec}
\end{align}
where $\rho_1= \frac{\lambda_{\min}(\mathbf\Theta)}{\rho}$. Again, the definition of $\mathbf{e}_f$ yields $\Vert \mathbf{e}_f \Vert \geq \Vert \mathbf e_1 \Vert$. Then, for a $\underline{\rho} \in \mathbb{R}^{+}$ one has
\begin{align}
\underline{\rho} \Vert \mathbf{e}_f \Vert^2 &\leq  V_f \leq \rho \Vert \mathbf{e}_f \Vert^2 \nonumber\\
\Rightarrow \Vert \mathbf e_1 \Vert &\leq \sqrt{{V_f}/{\underline{\rho}}} ~~\forall t.
\end{align}
Let $\Vert \mathbf e_1 (t_{in}+t_1) \Vert =\phi  $. Then $V_f  (t_{in}+t_1) \geq \underline{\rho} \phi^2 $. Since $V_f$ exponentially decreases for $t \geq  t_{in}+t_1$, there exists a finite time $\delta t_1= t- (t_{in}+t_1)$ such that $V_f(t_{in}+t_1+\delta t_1)= \underline{\rho} \phi^2$ implying $\Vert \mathbf e_1 (t_{in}+t_1+\delta t_1) \Vert < \phi  $. Thus $\mathbf{e}_1 \dot{\mathbf e}_1 \leq 0$ would occur and $\Vert \mathbf e_1 \Vert$ will stop increasing. From the definition of $\mathbf s$, it can be inferred that $\hat{c}$ would start decreasing following (\ref{ATRC}) for $t \geq t_{in}+T_1$ where $T_1 \leq t_1+\delta t_1$. The time $\delta t_1$ can be found from (\ref{v dec}) as
\begin{align}
\underline{\rho} \phi^2 & \leq {V}_f(t_{in}+t_1)e^{- \rho_1 \delta t_1} \nonumber \\
\Rightarrow \delta t_1 & \leq \frac{1}{\rho_1}\text{ln} \left(\frac{{V}_f(t_{in}+t_1)}{\underline{\rho} \phi^2} \right).
\end{align}
As $\Vert \mathbf{s} \Vert \leq \Vert \hat{\mathbf{s}} \Vert$, one can verify for that $\Vert \mathbf{s} \Vert \geq \epsilon$ 
\begin{equation}
 \hat{c}(t) \leq \frac{(\Upsilon+c+\bar{c} \delta t_1)\Vert \iota_1 \Vert \varpi_0+\Gamma/2}{\alpha \varpi }:=\hat{c}_M. \label{appen bound0}
\end{equation}

Now consider $\Vert \mathbf{s} \Vert < \epsilon$. From (\ref{ATRC}) and (\ref{vf dot}), one has
\begin{align}
	\dot{V}_f&\leq -{\mathbf e}_f^T {\mathbf\Theta \mathbf e_f}+\Gamma+2{\hat{\mathbf s}}^T(-\alpha \hat{c} {\mathbf{s}}/{\epsilon}+\boldsymbol \sigma_1) +2{\hat{\mathbf s}}^T \Upsilon\nonumber\\
	&\leq -{\mathbf e}_f^T {\mathbf\Theta \mathbf e_f}+\Gamma-2\alpha \hat{c}{\Vert\mathbf{s}\Vert^2}/{\epsilon} +2(\Upsilon+c)\Vert \iota_1 \Vert \Vert{\mathbf e}_f\Vert \nonumber\\
	& \leq - \lambda_{\min}{(\mathbf\Theta)\Vert\mathbf e_f\Vert^2}-2\alpha \hat{c}{\varpi^2}/{\epsilon} +2(\Upsilon+c)\Vert \iota_1 \Vert \Vert{\mathbf e}_f\Vert+\Gamma
\end{align}
Thus, similar to the earlier arguments, $\dot{V}_f<0$ when 
\begin{equation}
 \alpha \frac{\varpi^2}{\epsilon} \hat{c} \geq (\Upsilon+c)\Vert \iota_1 \Vert \varpi_0+\Gamma/2. \label{appen bound1}
\end{equation}
Let $\bar{t}_{in}$ be any arbitrary initial time when $\hat{c}$ starts increasing for $\Vert \mathbf{s} \Vert < \epsilon$. Then integrating both sides of (\ref{c low}) and using (\ref{appen bound1}) one can find that there exist finite time $t_2$ such that 
\begin{align}
t_2 \leq \frac{\epsilon\left( (\Upsilon+c)\Vert \iota_1 \Vert \varpi_0+\Gamma/2\right)}{\alpha \bar{c} \varpi^3}.
\end{align}
Then following the exact procedure derived for $\Vert \mathbf{s} \Vert < \epsilon$, one can find a finite time $\delta t_2=t-(\bar{t}_{in}+t_2)$ such that $\hat{c}$ starts to decrease following (\ref{ATRC}) for $t \geq \bar{t}_{in}+T_2$ where $T_2 \leq t_2+\delta t_2$. Thus, the following results can be obtained:
\begin{align}
 \delta t_2 & \leq \frac{1}{\rho_1}\text{ln} \left(\frac{{V}_f(\bar{t}_{in}+t_2)}{\underline{\rho} \Vert \mathbf e_1 (\bar{t}_{in}+t_2) \Vert^2} \right). \\
 \hat{c}(t) & \leq \frac{\epsilon \left \lbrace (\Upsilon+c+\bar{c} \delta t_2)\Vert \iota_1 \Vert \varpi_0+\Gamma/2 \right \rbrace}{\alpha \varpi^2 }:=\hat{c}_m. \label{appen bound2}
\end{align}
Hence $\hat{c}(t) \leq \max \lbrace \hat{c}_M, ~\hat{c}_m \rbrace :=c^{*}$.

	\textbf{(3) Derivation of $\dot{V}_r$}:
	Evaluating the structures of $\mathbf{s}$, $\Delta \mathbf{s}$ and the condition (\ref{c2}) we have
	\begin{align}
	 \Delta \mathbf{s}^T \mathbf{s} & > \mathbf 0. \label{new 1}
	\end{align}
	Further, the first condition of (\ref{ATRC}) and the initial condition $\hat{c}(t_0) > \gamma$ implies $\hat{c}(t) \geq \gamma$ $\forall t \geq t_0$. Then, by using (\ref{tarc lya}) and (\ref{vf dot}), the stability analysis for (\ref{sys new}) employing TARC is carried out for 
	all the possible four cases as follows:\\
	\par 
	\noindent \textbf{Case (i):} $ \Vert\mathbf{ s} \Vert\geq \epsilon \land \lbrace \hat{c} \leq \gamma \lor f(\mathbf{e}_1)>0\rbrace$.
	
	\indent Utilizing (\ref{delta u}), (\ref{ATRC}) and (\ref{vf dot}), we have 
	\begin{align}
	\dot{V}_r&\leq -{\mathbf e}_f^T {\mathbf\Theta \mathbf e_f}+\Gamma+2{\hat{\mathbf s}}^T(-\alpha \hat{c} {\mathbf{s}}/{\Vert\mathbf{ s} \Vert}+\boldsymbol \sigma_1) +2{\hat{\mathbf s}}^T \Upsilon\nonumber\\
	& \quad +2(\hat{c}-c)\Vert\mathbf{s}\Vert \nonumber\\
	&\leq -{\mathbf e}_f^T {\mathbf\Theta \mathbf e_f}+\Gamma-2(\alpha-1) \hat{c}\Vert\mathbf{s}\Vert-2\alpha \hat{c} \frac{\Delta \mathbf{s}^T\mathbf{s}}{\Vert\mathbf{ s} \Vert}\nonumber \\
	& \quad +2(\Upsilon+c) \Vert{\hat{\mathbf s}}\Vert \nonumber\\
	& \leq -\lambda_{\min}{(\mathbf\Theta)\Vert\mathbf e_f\Vert^2}+2 (\Upsilon+c) \Vert\boldsymbol{\iota}_1\Vert \Vert{\mathbf e_f}\Vert+\Gamma \nonumber, \label{case 1}
	\end{align} 
	\noindent since $\alpha$ is selected as $\alpha \geq 1$ and (\ref{c2}), (\ref{new 1}) hold. $\exists \rho,c^{*} \in \mathbb{R}^{+}$ such that the definition of $V_r$ yields
	\begin{align}
	V_r &\leq \rho \Vert\mathbf e_f\Vert^2+ \zeta, 
    \end{align}	 
	where $\zeta=\frac{{c^{*}}^2+c^2}{\bar{c}}$ and $\hat{c}(t) \leq c^{*}$. Then
	\begin{align}
	\dot{V}_r  \leq & -z \Vert\mathbf e_f \Vert^2 -(\lambda_{\min}{(\mathbf\Theta)-z)\Vert\mathbf e_f\Vert^2}+2 (\Upsilon+c) \Vert\boldsymbol{\iota}_1\Vert \Vert{\mathbf e_f}\Vert\nonumber\\
	&+\Gamma \nonumber\\
	\leq & - \rho_m V_r -(\lambda_{\min}{(\mathbf\Theta)-z)\Vert\mathbf e_f\Vert^2}+2 (\Upsilon+c) \Vert\boldsymbol{\iota}_1\Vert \Vert{\mathbf e_f}\Vert\nonumber\\
	&+\Gamma+ \rho_m \zeta, \nonumber
	\end{align}
	where $\rho_m=z/ \rho$. $\dot{V}_r<0$ is 
	established if $(\lambda_{\min}{(\mathbf\Theta)-z)\Vert\mathbf e_f\Vert^2}>\Gamma+\rho_m \zeta+ 2 (\Upsilon+c) \Vert\boldsymbol{\iota}_1\Vert \Vert{\mathbf e_f}\Vert$.
	Thus, the system 
	is UUB. 
	\par 
	\noindent \textbf{Case (ii):} $\Vert \mathbf{s} \Vert\geq \epsilon \land  f(\mathbf{e}_1) \leq 0$. 
	\begin{align}
	\dot{V}_r&\leq -{\mathbf e}_f^T {\mathbf\Theta \mathbf e_f}+\Gamma+2{\hat{\mathbf s}}^T(-\alpha \hat{c} {\mathbf{s}}/{\Vert \mathbf{s} \Vert}+\boldsymbol{\sigma}_1) +2{\hat{\mathbf s}}^T \Upsilon\nonumber\\
	& \quad - 2({\underline{c}}/{\bar{c}})(\hat{c}-c)\Vert\mathbf{s}\Vert\nonumber\\
	& \leq -\lambda_{\min}({\mathbf\Theta)\Vert\mathbf e_f\Vert^2}+2(\Upsilon+ (1+({\underline{c}}/{\bar{c}}))c) \Vert\boldsymbol{\iota}_1\Vert \Vert{\mathbf e_f}\Vert+\Gamma \nonumber\\
	& \leq  - \rho_m V_r -(\lambda_{\min}{(\mathbf\Theta)-z)\Vert\mathbf e_f\Vert^2}\nonumber\\
	&\quad +2(\Upsilon+ (1+({\underline{c}}/{\bar{c}}))c) \Vert\boldsymbol{\iota}_1\Vert \Vert{\mathbf e_f}\Vert+\Gamma+ \rho_m \zeta \nonumber.
	\label{case 2}
	\end{align}
	\noindent $\dot{V}_r<0$ 
	is achieved if $(\lambda_{\min}{(\mathbf\Theta)-z)\Vert\mathbf e_f\Vert^2}>\Gamma+\rho_m \zeta+ 2(\Upsilon+ (1+({\underline{c}}/{\bar{c}}))c) \Vert\boldsymbol{\iota}_1\Vert \Vert{\mathbf e_f}\Vert$; and the system is 
	UUB. 
	\begin{equation*}
	\end{equation*}
	\par 
	\noindent \textbf{Case (iii):} $\Vert \mathbf{s} \Vert < \epsilon \land \lbrace \hat{c} \leq \gamma \lor f(\mathbf{e}_1)>0\rbrace $.
	
	Since $\Vert\mathbf{s}\Vert<\epsilon$ and $\hat{c} \leq c^{*}$, one has
	\begin{align}
	\dot{V}_r&\leq -{\mathbf e}_f^T {\mathbf\Theta \mathbf e_f}+\Gamma+2{\hat{\mathbf s}}^T(-\alpha \hat{c} {\mathbf{s}}/{\epsilon}+\boldsymbol{\sigma_1})+2{\hat{\mathbf s}}^T \Upsilon\nonumber\\
	&\quad +2(\hat{c}-c)\Vert\mathbf{s}\Vert\nonumber\\
	& \leq -\lambda_{\min}{(\mathbf\Theta)\Vert\mathbf e_f\Vert^2}-(2\alpha\hat{c}{\Vert\mathbf{s}\Vert^2}/{\epsilon })-2\alpha\hat{c}\frac{\Delta \mathbf{s}^T \mathbf{s}}{\epsilon }+\Gamma\nonumber\\
	& \quad +2(\Upsilon+c) \Vert{\hat{\mathbf s}}\Vert+2\hat{c}\Vert\mathbf{s}\Vert \nonumber\\
	& \leq  - \rho_m V_r -(\lambda_{\min}{(\mathbf\Theta)-z)\Vert\mathbf e_f\Vert^2} +2 (\Upsilon+c) \Vert\boldsymbol{\iota}_1\Vert \Vert{\mathbf e_f}\Vert \nonumber \\
	&\quad +\Gamma+ \rho_m \zeta+ 2\epsilon c^{*}. \label{case 41}
	\end{align}
%
	\noindent  Thus, the system is UUB. 
	%
	%
	\par 
	\noindent\textbf{Case (iv):}  $\Vert \mathbf{s} \Vert < \epsilon \land  f(\mathbf{e}_1) \leq 0$. 
	\begin{align}
	\dot{V}_r& \leq -{\mathbf e}_f^T {\mathbf\Theta \mathbf e_f}+\Gamma+2{\hat{\mathbf s}}^T(-\alpha \hat{c} {\mathbf{s}}/{\epsilon}+\boldsymbol{\sigma}_1) +2{\hat{\mathbf s}}^T \Upsilon\nonumber\\
	& \quad - 2({\underline{c}}/{\bar{c}})(\hat{c}-c)\Vert\mathbf{s}\Vert\nonumber\\
	& \leq -\lambda_{\min}({\mathbf\Theta)\Vert\mathbf e_f\Vert^2}+2(\Upsilon+ (1+({\underline{c}}/{\bar{c}}))c) \Vert\boldsymbol{\iota}_1\Vert \Vert{\mathbf e_f}\Vert+\Gamma. \label{case 5}
	\end{align}
	Stability of this condition is exactly similar to Case (ii) and therefore, the system remains UUB.
%
%
%
\end{proof}

\subsection{Selection of Parameters} \label{sec:parameter}
{For stability, one needs to select the parameters $\mathbf Q,{\mathbf K_1, \mathbf K_2}, h,\varsigma, \alpha, \beta,\xi, \mathbf{D}$ and $\mathbf{L}$ such that (\ref{c2}) and (\ref{delay value 1}) are satisfied. Amongst them, $\beta,\xi, \mathbf{D}$ and $\mathbf{L}$ are solely used for analytical purpose. Note that there are many possible combinations of parameters which can satisfy the aforementioned conditions. Henceforth, a designer has the flexibility to select any of the combinations according to the application requirements. Nevertheless, a certain design procedure can be considered while selecting the parameters as follows: 
	
	As equation (\ref{sys}) represents the second order system, the controller gains $\mathbf K_1,\mathbf K_2$ are generally selected as $\mathbf K_1=\omega_n^2 \mathbf I$ and $\mathbf K_2=2 \zeta \omega_n \mathbf I$, where $\omega_n$ and $\zeta$ are the desired natural frequency and damping ratio, respectively, for the unperturbed (or nominal) error dynamics \cite{Ref:7,Ref:9,Ref:9_1}. While the designer can choose any $\zeta$ and $\omega_n$ according to the requirement, $\zeta=1$ is generally selected to achieve a critical damping behaviour of the nominal error dynamics \cite{Ref:7,Ref:9,Ref:9_1}. The sampling time $h$ is set to the minimum achievable value in a digital controller, e.g., $h = 1$ms, thus always known a priori. After selecting these parameters, rest of the parameters, i.e., $\varsigma,\beta,\xi, \mathbf{D,L}$ are to be selected in a way such that (\ref{delay value 1}) is satisfied.}

{ Apart from the aforementioned parameters, the two important user defined scalars $\bar{c}$ and $\underline{c}$ in (\ref{ATRC}) govern the adaptation rate of $\hat{c}$ while increasing and decreasing, respectively. Note that for system stability, it is sufficient to select $\bar{c}>0$ and $\underline{c}>0$. The larger values of $\bar{c}$ allow $\hat{c}$ to respond more quickly according to the incurred error. However, if $\bar{c}$ is too high then $\hat{c}$ becomes excessively high which may in turn invite chattering. Similarly, if $\underline{c}$ is too high then $\hat{c}$ becomes excessively low which may deteriorate the tracking accuracy due to the low value of switching gain than the required amount. On the other hand, a very low value of $\bar{c}$ disrupts $\hat{c}$ to counter the uncertainties properly and may result in a high tracking error. In addition, a very small value of $\underline{c}$ results in higher value of $\hat{c}$ than the sufficient amount resulting high control input. Thus, one needs to select these parameters according to the application requirements in practical circumstances.}

\section{Conclusion}
{In this endeavour, a novel adaptive-robust control law, TARC has been proposed for uncertain Euler-Lagrange systems which approximates unknown dynamics through the time-delayed estimation technique and negotiates the approximation error by adaptive-robust control without any prior knowledge of the uncertainty bounds. The proposed adaptive law of TARC overcomes the over- and under-estimation problems of the switching gain. Additionally, to enhance the practicality of TARC, the velocity and acceleration feedback terms are estimated from the previous position information; thus only position information is sufficient for the overall controller design. A new stability approach in the continuous-time domain analyses the overall closed-loop system including the TARC and the state-derivatives estimator. Moreover, in-depth discussion on selections of the controller gains and sampling interval (time delay) is established via the stability analysis, which is of importance for the time-delayed estimation based controllers. 


\begin{thebibliography}{10}
\providecommand{\url}[1]{#1}
\csname url@samestyle\endcsname
\providecommand{\newblock}{\relax}
\providecommand{\bibinfo}[2]{#2}
\providecommand{\BIBentrySTDinterwordspacing}{\spaceskip=0pt\relax}
\providecommand{\BIBentryALTinterwordstretchfactor}{4}
\providecommand{\BIBentryALTinterwordspacing}{\spaceskip=\fontdimen2\font plus
\BIBentryALTinterwordstretchfactor\fontdimen3\font minus
  \fontdimen4\font\relax}
\providecommand{\BIBforeignlanguage}[2]{{%
\expandafter\ifx\csname l@#1\endcsname\relax
\typeout{** WARNING: IEEEtran.bst: No hyphenation pattern has been}%
\typeout{** loaded for the language `#1'. Using the pattern for}%
\typeout{** the default language instead.}%
\else
\language=\csname l@#1\endcsname
\fi
#2}}
\providecommand{\BIBdecl}{\relax}
\BIBdecl

\bibitem{Ref:2}
X.~Liu, H.~Su, B.~Yao, and J.~Chu, ``Adaptive robust control of a class of
  uncertain nonlinear systems with unknown sinusoidal disturbances,'' in
  \emph{47th IEEE Conf. on Decis. and Control}.\hskip 1em plus 0.5em minus
  0.4em\relax IEEE, 2008, pp. 2594--2599.

\bibitem{Ref:3}
M.~Corless and G.~Leitmann, ``Continuous state feedback guaranteeing uniform
  ultimate boundedness for uncertain dynamic systems,'' \emph{IEEE Trans.
  Autom. Control}, vol.~26, no.~5, pp. 1139--1144, 1981.

\bibitem{Ref:4}
H.~Lee and V.~I. Utkin, ``Chattering suppression methods in sliding mode
  control systems,'' \emph{Annual reviews in control}, vol.~31, no.~2, pp.
  179--188, 2007.

\bibitem{Ref:13}
X.~Zhu, G.~Tao, B.~Yao, and J.~Cao, ``Adaptive robust posture control of
  parallel manipulator driven by pneumatic muscles with redundancy,''
  \emph{IEEE/ASME Trans. Mechatron.}, vol.~13, no.~4, pp. 441--450, 2008.

\bibitem{Ref:16}
W.~Sun, Z.~Zhao, and H.~Gao, ``Saturated adaptive robust control for active
  suspension systems,'' \emph{IEEE Trans. Ind. Electron.}, vol.~60, no.~9, pp.
  3889--3896, 2013.

\bibitem{Ref:17}
S.~Islam, P.~X. Liu, and A.~El~Saddik, ``Robust control of four-rotor unmanned
  aerial vehicle with disturbance uncertainty,'' \emph{IEEE Trans. Ind.
  Electron.}, vol.~62, no.~3, pp. 1563--1571, 2015.

\bibitem{Ref:18}
Z.~Chen, B.~Yao, and Q.~Wang, ``$\mu$-synthesis-based adaptive robust control
  of linear motor driven stages with high-frequency dynamics: A case study,''
  \emph{IEEE/ASME Trans. Mechatron.}, vol.~20, no.~3, pp. 1482--1490, 2015.

\bibitem{Ref:19}
Z.~Liu, H.~Su, and S.~Pan, ``A new adaptive sliding mode control of uncertain
  nonlinear systems,'' \emph{Asian J. Control}, vol.~16, no.~1, pp. 198--208,
  2014.

\bibitem{Ref:20}
C.-Y. Chen, T.-H.~S. Li, Y.-C. Yeh, and C.-C. Chang, ``Design and
  implementation of an adaptive sliding-mode dynamic controller for wheeled
  mobile robots,'' \emph{Mechatronics}, vol.~19, no.~2, pp. 156--166, 2009.

\bibitem{Ref:20-1}
A.~K. Khalaji and S.~A.~A. Moosavian, ``Robust adaptive controller for a
  tractor--trailer mobile robot,'' \emph{IEEE/ASME Trans. Mechatron.}, vol.~19,
  no.~3, pp. 943--953, 2014.

\bibitem{Ref:20-2}
Q.~Meng, T.~Zhang, X.~Gao, and J.-y. Song, ``Adaptive sliding mode
  fault-tolerant control of the uncertain stewart platform based on offline
  multibody dynamics,'' \emph{IEEE/ASME Trans. Mechatron.}, vol.~19, no.~3, pp.
  882--894, 2014.

\bibitem{Ref:20-3}
A.~Nasiri, S.~K. Nguang, and A.~Swain, ``Adaptive sliding mode control for a
  class of mimo nonlinear systems with uncertainties,'' \emph{J. Franklin
  Inst.}, vol. 351, no.~4, pp. 2048--2061, 2014.

\bibitem{Ref:32}
S.~Liu, H.~Zhou, X.~Luo, and J.~Xiao, ``Adaptive sliding fault tolerant control
  for nonlinear uncertain active suspension systems,'' \emph{J. Franklin
  Inst.}, vol. 353, no.~1, pp. 180--199, 2016.

\bibitem{Ref:21}
F.~Plestan, Y.~Shtessel, V.~Bregeault, and A.~Poznyak, ``New methodologies for
  adaptive sliding mode control,'' \emph{Int. J. control}, vol.~83, no.~9, pp.
  1907--1919, 2010.

\bibitem{Ref:22}
------, ``Sliding mode control with gain adaptation--application to an
  electropneumatic actuator,'' \emph{Control Eng. Pract.}, vol.~21, no.~5, pp.
  679--688, 2013.

\bibitem{Ref:33}
S.~Roy, S.~Nandy, R.~Ray, and S.~N. Shome, ``Robust path tracking control of
  nonholonomic wheeled mobile robot: Experimental validation,'' \emph{Int. J.
  Control Autom. Syst.}, vol.~13, no.~4, pp. 897--905, 2015.

\bibitem{Ref:6}
T.~Hsia and L.~Gao, ``Robot manipulator control using decentralized linear
  time-invariant time-delayed joint controllers,'' in \emph{proc. IEEE Int.
  Conf. Robot. Autom.}\hskip 1em plus 0.5em minus 0.4em\relax IEEE, 1990, pp.
  2070--2075.

\bibitem{Ref:6-2}
K.~Youcef-Toumi and O.~Ito, ``A time delay controller for systems with unknown
  dynamics,'' \emph{ASME J. Dyn. Syst. Meas. Control}, vol. 112, no.~1, pp.
  133--142, 1990.

\bibitem{Ref:26}
Y.-H. Shin and K.-J. Kim, ``Performance enhancement of pneumatic vibration
  isolation tables in low frequency range by time delay control,'' \emph{J.
  Sound Vib.}, vol. 321, no.~3, pp. 537--553, 2009.

\bibitem{Ref:new1}
J.~Lee \emph{et~al.}, ``An experimental study on time delay control of
  actuation system of tilt rotor unmanned aerial vehicle,''
  \emph{Mechatronics}, vol.~22, no.~2, pp. 184--194, 2012.

\bibitem{Ref:7}
G.~R. Cho, P.~H. Chang, S.~H. Park, and M.~Jin, ``Robust tracking under
  nonlinear friction using time-delay control with internal model,'' \emph{IEEE
  Trans. Control Syst. Technol.}, vol.~17, no.~6, pp. 1406--1414, 2009.

\bibitem{Ref:9}
M.~Jin, S.~H. Kang, and P.~H. Chang, ``Robust compliant motion control of robot
  with nonlinear friction using time-delay estimation,'' \emph{IEEE Trans. Ind.
  Electron.}, vol.~55, no.~1, pp. 258--269, 2008.

\bibitem{Ref:9_1}
J.~Lee, P.~H. Chang, and R.~S. Jamisola, ``Relative impedance control for
  dual-arm robots performing asymmetric bimanual tasks,'' \emph{IEEE Trans.
  Indus. Electron.}, vol.~61, no.~7, pp. 3786--3796, 2014.

\bibitem{Ref:10}
Y.~Jin, P.~H. Chang, M.~Jin, and D.~G. Gweon, ``Stability guaranteed time-delay
  control of manipulators using nonlinear damping and terminal sliding mode,''
  \emph{IEEE Trans. Indus. Electron.}, vol.~60, no.~8, pp. 3304--3317, 2013.

\bibitem{Ref:12}
S.~Roy, S.~Nandy, R.~Ray, and S.~N. Shome, ``Time delay sliding mode control of
  nonholonomic wheeled mobile robot: experimental validation,'' in \emph{proc.
  IEEE Int. Conf. Robot. Autom.}\hskip 1em plus 0.5em minus 0.4em\relax IEEE,
  2014, pp. 2886--2892.

\bibitem{Ref:34}
J.~Kim, H.~Joe, S.-c. Yu, J.~S. Lee, and M.~Kim, ``Time-delay controller design
  for position control of autonomous underwater vehicle under disturbances,''
  \emph{IEEE Trans. Ind. Electron.}, vol.~63, no.~2, pp. 1052--1061, 2016.

\bibitem{Ref:faltu}
T.~Das and I.~N. Kar, ``Design and implementation of an adaptive fuzzy
  logic-based controller for wheeled mobile robots,'' \emph{IEEE Trans. Control
  Syst. Technol.}, vol.~14, no.~3, pp. 501--510, 2006.

\bibitem{Ref:jin}
M.~Jin, J.~Lee, and K.~K. Ahn, ``Continuous nonsingular terminal sliding-mode
  control of shape memory alloy actuators using time delay estimation,''
  \emph{IEEE/ASME Trans. Mechatron.}, vol.~20, no.~2, pp. 899--909, 2015.

\bibitem{Ref:6-1}
A.~G. Ulsoy, ``Time-delayed control of siso systems for improved stability
  margins,'' \emph{ASME J. Dyn. Syst. Meas. Control}, vol. 137, no.~4, pp.
  1--12, 2015.

\bibitem{Ref:dis_tdc_1}
J.~Jung, P.~Chang, and D.~Stefanov, ``Discretisation method and stability
  criteria for non-linear systems under discrete-time time delay control,''
  \emph{IET control theory \& applications}, vol.~5, no.~11, pp. 1264--1276,
  2011.

\bibitem{Ref:dis_TDC_2}
J.~Lee, G.~A. Medrano-Cerda, and J.~H. Jung, ``Corrections for ``discretisation
  method and stability criteria for non-linear systems under discrete-time time
  delay control'','' \emph{IET Control Theory \& Applications}, vol.~10,
  no.~14, pp. 1751--1754, 2016.

\bibitem{Ref:30}
J.~Reger and J.~Jouffroy, ``On algebraic time-derivative estimation and
  deadbeat state reconstruction,'' in \emph{proc. 48th IIEEE Conf. Decis.
  Control}.\hskip 1em plus 0.5em minus 0.4em\relax IEEE, 2009, pp. 1740--1745.

\bibitem{Ref:31}
K.~Gu, J.~Chen, and V.~L. Kharitonov, \emph{Stability of time-delay
  systems}.\hskip 1em plus 0.5em minus 0.4em\relax Springer, 2003.

\bibitem{Ref:24}
H.~K. Khalil, ``Nonlinear systems, 3rd,'' \emph{New Jewsey, Prentice Hall},
  vol.~9, 2002.

\end{thebibliography}
\end{document}